\documentclass[12pt]{article}
\usepackage{amsmath}
\usepackage{graphicx,psfrag,epsf}
\usepackage{enumerate}
\usepackage[figuresright]{rotating}
\usepackage{natbib}
\usepackage{amssymb}
\usepackage{pifont}
\usepackage{multirow}
\usepackage{booktabs}
\usepackage{url} 
\usepackage{adjustbox}
\usepackage{booktabs}
\usepackage{threeparttable}

\usepackage{amsmath}
\usepackage{graphicx,psfrag,epsf}
\usepackage{enumerate}
\usepackage[figuresright]{rotating}
\usepackage{natbib}
\usepackage{amssymb}
\usepackage{pifont}
\usepackage{multirow}
\usepackage{booktabs}
\usepackage{url} 
\usepackage{adjustbox}
\newcommand{\blind}{0}

\newcommand\uu{\mathbf{u}}

\newcommand\V{\mathbf{V}}
\newcommand\y{\mathbf{y}}
\newcommand\Z{\mathbf{Z}}

\newcommand\PP{\mathbf{P}}
\newcommand\X{\mathbf{X}}
\newcommand\Y{\mathbf{Y}}

\newcommand\PPi{\boldsymbol{\Pi}}
\newcommand\bbeta{\boldsymbol{\beta}}

\usepackage{amsmath}
\usepackage{tcolorbox}
\usepackage{subfig}
\usepackage{amssymb}
\usepackage{amsfonts}
\usepackage{multirow}
\usepackage{amsthm}
\usepackage{mathrsfs}
\usepackage{caption}
\usepackage{indentfirst}
\newtheorem{theorem}{Theorem}
\newtheorem{lemma}{Lemma}
\newtheorem{corollary}{Corollary}

\theoremstyle{definition}

\newtheorem{remark}{Remark}
\newtheorem{assumption}{Assumption}
\usepackage{fancybox,framed}

\usepackage[utf8]{inputenc}
\usepackage{graphicx}

\usepackage{float}
\usepackage[textsize=tiny]{todonotes}

\begin{document}

\def\spacingset#1{\renewcommand{\baselinestretch}%
{#1}\small\normalsize} \spacingset{1}


\if0\blind
{
  \title{\bf A specification test for the strength of instrumental variables}
  \author{Zhenhong Huang\footnote{Department of Statistics and Actuarial Science, The University of Hong Kong. Email: zhhuang7@connect.hku.hk}, \;
   Chen Wang\footnote{Department of Statistics and Actuarial Science, The University of Hong Kong. Email: stacw@hku.hk} \;and\; Jianfeng Yao\footnote{School of Data Science, The Chinese University of Hong Kong (Shenzhen). Email:jeffyao@cuhk.edu.cn}\hspace{.2cm} \\
    }
    \date{}
  \maketitle
} \fi

\if1\blind
{
  \bigskip
  \bigskip
  \bigskip
  \begin{center}
    {\LARGE\bf Title}
\end{center}
  \medskip
} \fi

\bigskip
\begin{abstract}
\noindent This paper develops a new specification test for the instrument weakness when the number of instruments $K_n$ is large  with a magnitude comparable to the  sample size $n$. The test relies on the fact that the difference between the two-stage least squares (2SLS) estimator   and the ordinary least squares (OLS) estimator asymptotically disappears when there are  many weak instruments, but otherwise converges to a non-zero limit. We establish the limiting distribution of the difference within the above two specifications, and introduce a delete-$d$ Jackknife procedure to consistently estimate the asymptotic variance/covariance of the difference. Monte Carlo experiments demonstrate the good performance of the test procedure for both cases of single and multiple endogenous variables. Additionally, we re-examine the analysis of returns to education data in \cite{angrist1991does}  using our proposed test. Both the simulation results and empirical analysis indicate the reliability of the test.  
\end{abstract}

\noindent%
\vfill

\newpage
\spacingset{1.8} 
\section{Introduction}
\label{sec:intro}
 In some instrumental variable (IV) models, empirical researchers often face situations where the number of instruments is large with a magnitude comparable to the sample size \citep{arellano1991some, dagenais1997higher, bhuller2020incarceration}. Instead of gaining efficiency as the conventional asymptotic theory indicates, the increasing number of instruments would deteriorate various IV estimators and test statistics \citep{bekker1994alternative, han2006gmm, anatolyev2011specification}. This is the first strand of theoretical econometric literature on the IV estimator, the so-called ``many instruments problem". 

The second strand focuses on the statistical properties of the related estimation and inference methods, the so-called ``weak instrument problem", where the instruments are weakly correlated with the endogenous variables. It is well understood that weakness in instruments will lead to bias and inconsistency of the IV estimators and test statistics \citep{bound1995problems,staiger1994instrumental}. The key quantity here is the concentration parameter\footnote{When there are multiple endogenous variables, the key quantity is the concentration matrix, whose eigenvalues contain the strength of instruments. For simplicity of illustration, we use the concentration parameter as a unified term. }   \citep{rothenberg1984approximating}, which depicts the strength of the instrument set. The conventional asymptotic setting assumes that the concentration parameter grows at the same rate as the sample size. However, in the weak instrument asymptotics, this parameter is typically of smaller order than the sample size. For example, \cite{staiger1994instrumental} proposed the local-to-zero framework to keep the concentration parameter roughly constant against the sample size and show that the classical IV estimators are not consistent. More generally, \cite{chao2005consistent} introduced a unified framework to amalgamate the aforementioned two strands by modeling the concentration parameter with an arbitrary and divergent sequence. They show that the consistent estimation is still feasible when the concentration parameter is of higher order than $\sqrt{K_n}$. From the perspective of inference, \cite{mikusheva2021inference} established that no asymptotically consistent test for the interested coefficients exists if the ratio of the concentration parameter over $\sqrt{K_n}$ stays bounded. Therefore, in applications involving many instruments, it raises a natural question of how to determine which range the concentration parameter is in.

Compared with developing estimation and inference methods that are robust to the many weak instruments problem, the literature on assessing the strength of instruments is exceptionally scarce. When the number of instruments is fixed, \cite{stock2002testing} proposed to use a first-stage $F$ test for the single endogenous variable and the minimum eigenvalue of the Cragg-Donald statistics for multiple endogenous variables, which depict the magnitude of the concentration parameter. They have tabulated the non-standard critical values by simulation experiments, and a ``rule of thumb" becomes a commonplace pre-test: reject the null of weak instruments when this $F$ statistic exceeds 10.  However, the validity of the test remains unclear when the number of instruments is large, even though the authors have shown its reliability when the size of the instrument set grows at a much slower rate than the sample size ($K_n^4/n\rightarrow0$). Also, for the case of multiple endogenous variables, the tabulated critical values are known to be conservative as pointed out in \cite{stock2002testing}. Building on their work, \cite{sanderson2016weak} considered tests for the purposes of estimation and inference on one of multiple endogenous variables, but still keeping the number of instruments fixed. For the case of many instruments,  \cite{hahn2002new} proposed a test to examine the adequacy of the standard asymptotic result in IV regression models. They argue that if the test rejects the null, then weakness in instruments may arise. However, the test lacks power in detecting weak instruments \citep{hausman2005asymptotic}. \cite{lee2012hahn} proved that it is indeed a test for the exogeneity of the instruments. For the case of single endogenous variable and many instruments, \cite{mikusheva2021inference} defined the instrument set to be weak when the concentration parameter in bounded in $\sqrt{K_n}$, and developed a pre-test for the weakness of the instruments in a spirit of the first-stage $F$ test.

In this paper, our contribution is to propose a specification test to distinguish between two different cases where the concentration parameter is at the same rate or a smaller rate of $K_n$ within the many instruments setup. Determining the order of the concentration parameter relative to $K_n$ is essential as the asymptotic behavior of the IV estimators and related inference methods heavily rely on it \citep{chao2006asymptotic,anatolyev2011specification}.
Our proposed test can also be viewed as a follow-up procedure to the test proposed in \cite{mikusheva2021inference}: if their test concludes that the ratio of the concentration parameter over $\sqrt{K_n}$ is large, then one can apply our test to gain further insight into the employed instruments set by testing whether the concentration parameter is smaller than $K_n$ in order.

Unlike the tests in \cite{stock2002testing} and \cite{mikusheva2021inference} that can depict the magnitude of the concentration parameter, our test is in the same manner of the test proposed in \cite{hahn2002new}, where they proposed a procedure to test for strong instruments by means of comparing  forward and reverse 2SLS. Our test is based on the finding that the difference between the 2SLS and OLS estimators disappears under many weak instruments asymptotics but deviates from zero under many strong instruments asymptotics. As an important technical contribution, we establish the joint distribution of 
 seven quadratic/bilinear forms that make up the 2SLS and OLS estimators. This complex joint distribution has its own interest: it can be used for the study of other related estimator such as the B2SLS estimator \citep{nagar1959bias} in this context of many weak instruments. As the asymptotic covariance matrix has no explicit form,   
 to implement our test we introduce a subsampling technique and propose a delete-$d$ Jackknife covariance matrix estimator. Both the theory and Monte Carlo experiments show that the test is robust to the number of endogenous variables and the type of error distributions.

The paper is organised as follows. In Section \ref{sec:model} we introduce the model and assumptions, followed by a discussion of the order of the concentration parameter. In Section \ref{sec:test}, we first give the convergence results of the 2SLS and OLS estimators under many strong instruments and many weak instruments asymptotics in Section \ref{subsec:diff}. Then in Section \ref{subsec:clt}, we derive the limiting distributions of the difference of the 2SLS and OLS estimators under the above two specifications. In Section \ref{subsec:jack} and \ref{subsec:test}, we show that the delete-$d$ Jackknife is a valid estimation method for the asymptotic covariance matrix under the null and formalize the new test.  Section \ref{sec:mc} reports the results of Monte Carlo simulations, followed by an empirical analysis in Section \ref{sec:emp} where we re-exaimine the results in \cite{angrist1991does}. Section \ref{sec:con} provides concluding remarks and further discussion.

\textbf{Notations.} Throughout the paper, for a vector $\mathbf{a}$, $\mathbf{a}'$ denotes its transpose. For an $a\times b$ matrix $\mathbf{A}$, we denote by $\mathbf{A}_i$ its $i$-th row, by $\mathbf{A}(j)$ its $j$-th column, and by $A_{ij}$ its $(i,j)$ entry, so that $\mathbf{A}=[\mathbf{A}_1',\mathbf{A}_2',\dots,\mathbf{A}_{a}']'=[\mathbf{A}(1),\mathbf{A}(2),\dots,\mathbf{A}(b)]$.  In addition, $\|\cdot\|$ represents the Euclidean norm for a vector and the induced operator norm for a matrix. $\mathbf{I}_{a}$ and $\mathbf{0}_{a}$ are $a \times a$ identity and null matrices, with $\mathbf{0}_{a \times b}$ being $a \times b$ matrix of zeros; $\boldsymbol{J}^{ij}$ is the single-entry matrix with 1 at $(i, j)$ and zero elsewhere and $``\boldsymbol{j}^{i}$ is the single-entry vector with 1 at $i$-th entry and zero elsewhere. Finally, $``\stackrel{p}{\rightarrow} "$ denotes the convergence in probability and $``\stackrel{d}{\rightarrow}$ " the convergence in distribution.

\section{Model and assumptions} \label{sec:model}
Consider the following model:
\begin{equation}
    \y=\Y\bbeta+\uu,
\end{equation}
\begin{equation}
    \Y=\Z\PPi+\V,
\end{equation} where $\y$ and $\Y$ are, respectively, an $n\times 1$ vector and an $n \times p$ matrix of observations on the endogenous variables of the system, $\Z$ is an $n\times K_n$ matrix of observations on the $K_n$ instrumental variables, and $\uu$ and $\V$ are, respectively, an $n\times 1$ vector and an $n \times p$ matrix of random disturbances. Within this setup, we specify $\PPi \left(K_{n} \times p\right)$ to depend on $n$ in order to model the effects of having many weak instruments. In particular, the effect of many instruments can be examined by letting $K_{n} \rightarrow \infty$ as $n \rightarrow \infty$, while the effect of weak instruments can be accounted for by shrinking $\PPi$ toward a zero matrix as $n$ grows. We also drop other potential exogenous variables without loss of generality. The following assumptions are used in the sequel.

\begin{assumption}\label{assum:1}(a) As $n\rightarrow\infty$, $K_n/n \rightarrow\alpha  \in (0,1)$; and (b) there exists a non-decreasing sequence of positive real numbers $\{s_n\}$ such that $s_n/n\rightarrow\kappa\in [0,\infty)$, and $
\boldsymbol{\Theta}_n=\PPi'\Z'\Z\PPi/s_n {\rightarrow} \boldsymbol{\Theta}$ almost surely, for some $p\times p$ nonrandom positive definite matrix $\boldsymbol{\Theta}$.
\end{assumption} 

\begin{assumption}\label{assum:2}(a) $\Z$ is independent of $\uu$ and $\V$, and the $\Z_i$'s are independently and identically distributed (i.i.d.); (b) $(\V_i,u_i)'$ is i.i.d. with zero mean and the    covariance matrix of $\boldsymbol{\Sigma}$ satisfying that $\mathrm{Cov}(\V_i)=\boldsymbol{\Sigma}_{VV}$, $\mathrm{Var}(u_i)=\sigma_u^2$, and $\mathrm{Cov}(\V_i,u_i)=\boldsymbol{\Sigma}_{Vu}$; and (c) there exist some positive constants $C_1<\infty$ and $C_2<\infty$ such that $E(u_i^4)<C_1$ and $E(V_{ij}^4)<C_2$.  
\end{assumption} 

Assumption \ref{assum:1}(a) adopts the many instruments asymptotic framework. We are particularly interested in the case where $K_n$ and $n$ are of the same order. Given that the concentration matrix, $\boldsymbol{\Sigma}_{VV}^{-1/2}\PPi'\Z'\Z\PPi\boldsymbol{\Sigma}_{VV}^{-1/2}$, is a natural measure of the instrument strength, Assumption \ref{assum:1}(b) models the strength of instruments by the order of the magnitude of $s_n$. We focus on the cases where $s_n$ grows no faster than $n$. In general, the slower the divergence of $s_n$ is, the weaker the instruments are. Assumption \ref{assum:2}(a) assumes the instruments are valid and i.i.d., and Assumption \ref{assum:2}(b) and (c) require the error terms to be homoscedastic and have the finite fourth moment.

A key parameter in the discussion is $s_n$, which measures the strength of the instrument set. The parameter plays an essential role in IV regressions: For example, the consistent estimation and the reliable inference of $\bbeta $ crucially depend on this parameter.  To further clarify the discussion, we propose the following taxonomy of instruments in terms of their strength: 
\begin{itemize}
    \item[(a)] The instrument set is ``completely weak" if $\frac{s_n}{\sqrt{n}} \rightarrow \kappa_1 \in [0, \infty),$ as $n\rightarrow \infty$.
    \item[(b)] The instrument set is ``moderately weak" if $\frac{s_n}{\sqrt{n}} \rightarrow  \infty$, but $\frac{s_n}{n}\rightarrow 0$ as $n\rightarrow \infty$.
    \item[(c)] The instrument set is ``strong" if $\frac{s_n}{n}\rightarrow \kappa_0 \in (0, \infty)$ as $n\rightarrow \infty$.
\end{itemize}

In case (a), no findings have been reported on the consistent estimation of $\bbeta$ to date, and \cite{mikusheva2021inference} (hereafter referred to as MS2022) show that a consistent test is absent as well. In case (b), \cite{chao2005consistent} show that the 2SLS estimator is inconsistent but consistent estimation is still feasible by using estimators such as the B2SLS  and LIML \citep{anderson1949estimation} estimators. However, they have a slow convergence rate  $\frac{s_n}{\sqrt{n}}$ within this specification \citep{chao2006asymptotic, hansen2008estimation}. In this paper we call the set of instruments ``weak" if $s_n=o(n)$ by integrating the cases (a) and (b): instruments are weak as a group and it refers to the many-weak-instrument asymptotics \citep{hausman2012instrumental, bekker2015jackknife}.
In case (c), $s_n$ has the order of $n$ and the concentration parameter grows as fast as the sample size. This case corresponds to the standard many-strong-instrument asymptotics: instruments are many and
they are strong as a group \citep{donald2001choosing,anderson2010asymptotic,anatolyev2013instrumental}. In this case, the 2SLS estimator is still inconsistent, and the B2SLS and the LIML estimators remain consistent with the convergence rate $\sqrt{n}$. 

In empirical applications, it is of practical importance to distinguish between the case of $\frac{s_n}{n}\rightarrow 0$ and the case of $\frac{s_n}{n}\rightarrow \kappa_0>0$. Firstly, as mentioned above,  the B2SLS, LIML and JIVE \citep{angrist1995split} estimators have the normal convergence rate $\sqrt{n}$ in the latter case. However, in the former case, they have the slower convergence rate $\frac{s_n}{\sqrt{n}}$ with  moderately weak instruments and become inconsistent with completely weak instruments. For example, when  $s_nn^{\gamma+1/2}$ for a constant $0<\gamma<1/2$, the convergence rate is then $n^\gamma$, which can be slow (e.g. $\gamma=0.1$) even though consistency can still be achieved. Such slow convergence rate of the aforementioned estimators can lead to biased results in finite sample. Besides, some robust inference methods with many instruments require the instruments to be strong, i.e., they are developed under many-strong-instrument asymptotics. For example, \cite{lee2012hahn} modified the  overidentification $J$ test for testing the hypothesis of valid instruments and established the asymptotic normality the modified test statistic when $\frac{s_n}{n}\rightarrow \kappa_0$. Similar results can also be found in \cite{anatolyev2011specification}. If the instruments are not strong,  then one should be cautious about the usage of their proposed tests. Therefore, determining the choice of many-instrument aymptotics is crucial when applying such inference methods. We therefore consider testing the hypothesis
\begin{equation}
    H_0 : \frac{s_n}{n}\rightarrow 0 \quad \quad \text{v.s.}\quad \quad H_1 : \frac{s_n}{n}\rightarrow \kappa_0 \in(0,\infty). \nonumber
\end{equation}

\section{The new specification test} \label{sec:test}
In this section, we introduce our testing procedure to test for the strength of the available instrument set.
\subsection{Limits of 2SLS and OLS} \label{subsec:diff}
To motivate the new specification test, we firstly study  the limits of $\hat{\bbeta}^{2SLS}=(\Y'\PP_{Z}\Y)^{-1}$\\$\Y'\PP_{Z}\y$ and $\hat{\bbeta}^{OLS}=(\Y'\Y)^{-1}\Y'\y$ under many instruments asymptotics in the following.
\begin{theorem} Assume that Assumptions \ref{assum:1} and \ref{assum:2} hold. Then, as $n\rightarrow\infty$, 
(a) under $H_1$, $\hat{\bbeta}^{2SLS}  \stackrel{p}{\rightarrow} \bbeta+ (\kappa_0/\alpha\boldsymbol{\Theta}+\boldsymbol{\Sigma}_{VV})^{-1}\boldsymbol{\Sigma}_{Vu}$, $\hat{\bbeta}^{OLS} \stackrel{p}{\rightarrow} \bbeta+ (\kappa_0\boldsymbol{\Theta}+\boldsymbol{\Sigma}_{VV})^{-1}\boldsymbol{\Sigma}_{Vu}$; and (b) under $H_0$, both $\hat{\bbeta}^{2SLS}$ and $\hat{\bbeta}^{OLS}$ converge in probability to $ \bbeta+ \boldsymbol{\Sigma}_{VV}^{-1}\boldsymbol{\Sigma}_{Vu}$.\label{theo:limit} 
\end{theorem}
It has been well documented that 2SLS suffers from the  many instruments problem. Theorem \ref{theo:limit} provides a more quantitatively asymptotic analysis of the phenomenon. Part (a) shows that when the instruments are strong as a group, the asymptotic bias of the 2SLS and OLS estimators differs, and the bias of 2SLS is strictly smaller than the bias of OLS. As the number of instruments increases, the asymptotic bias of 2SLS gets more severe, which is also observed by \cite{bekker1994alternative}. When $\alpha$ is close to zero (a fixed number of instruments), 2SLS will apply with little bias, and when $\alpha$ approaches one, 2SLS is biased towards OLS \citep{buse1992bias}.  Part (b) indicates that this difference disappears in the presence of many weak instruments. We also find that the asymptotic bias of two estimators in (b) becomes greater compared with those in (a), suggesting that weakness in instruments exacerbates the estimation bias in many instruments asymptotics.

\subsection{Limiting distribution of the difference between 2SLS and OLS} \label{subsec:clt}
Inspired by Theorem \ref{theo:limit}, we adopt the ideas of the specification test approach of \cite{hausman1978specification} involving the 2SLS and OLS estimators and see how far apart they are. If the difference between the two estimators is small, one should not reject the asymptotics adopted in the mode. If the difference is large, one should come to the opposite conclusion. We firstly derive the limiting distribution of the difference $\hat{\bbeta}^{2SLS}-\hat{\bbeta}^{OLS}$ under both many weak and many strong asymptotics in the following theorem, 
\begin{theorem} \label{theo:stat}(a) Under $H_0$, Assumptions \ref{assum:1} and \ref{assum:2}, as $n\rightarrow\infty$,
\begin{equation} \label{null_variance}
\sqrt{n}\left( \hat{\bbeta}^{2SLS}- \hat{\bbeta}^{OLS} \right)\stackrel{d}{\rightarrow}N_p(\boldsymbol{0}, \boldsymbol{\Sigma}_{0}), 
\end{equation}
where $\boldsymbol{\Sigma}_{0}=-\boldsymbol{g}_3\boldsymbol{\Sigma}_{3}\boldsymbol{g}_3'+\frac{1}{\alpha^2}\boldsymbol{g}_3\boldsymbol{\Sigma}_{4}\boldsymbol{g}_3'$ with $\boldsymbol{g}_3$, $\boldsymbol{g}_4$, $ \boldsymbol{\Sigma}_{3}$ and $\boldsymbol{\Sigma}_{4}$ defined in (\ref{g3}),  (\ref{g4}),  (\ref{sigma3}) and (\ref{sigma4}), respectively.

(b) Under $H_1$, Assumptions \ref{assum:1} and \ref{assum:2},  as $n\rightarrow\infty$,
\begin{equation} \label{h1b}
\sqrt{n}\left( \hat{\bbeta}^{2SLS}- \hat{\bbeta}^{OLS}-\boldsymbol{\Delta} \right)\stackrel{d}{\rightarrow}N_p(\boldsymbol{0},\boldsymbol{\Sigma}_{A}), 
\end{equation}
where 
    $\boldsymbol{\Delta}=(\kappa_0/\alpha\boldsymbol{\Theta}+\boldsymbol{\Sigma}_{VV})^{-1}\boldsymbol{\Sigma}_{Vu}-(\kappa_0\boldsymbol{\Theta}+\boldsymbol{\Sigma}_{VV})^{-1}\boldsymbol{\Sigma}_{Vu}+O_p(n^{-1/2})$ and  $\boldsymbol{\Sigma}_{A}=\boldsymbol{h}_1\boldsymbol{\Sigma}_{1}\boldsymbol{h}_1'+\boldsymbol{h}_2\boldsymbol{\Sigma}_{2}\boldsymbol{h}_2'+\boldsymbol{h}_3\boldsymbol{\Sigma}_{3}\boldsymbol{h}_3'+\alpha\boldsymbol{h}_3\boldsymbol{\Sigma}_{3}\boldsymbol{h}_4'+\alpha\boldsymbol{h}_4\boldsymbol{\Sigma}_{3}\boldsymbol{h}_3'+\boldsymbol{h}_4\boldsymbol{\Sigma}_{4}\boldsymbol{h}_4'$ with $\boldsymbol{h}_1$, $\boldsymbol{h}_2$, $\boldsymbol{h}_3$, $\boldsymbol{h}_4$, $\boldsymbol{\Sigma}_{1}$, $\boldsymbol{\Sigma}_{2}$ defined in (\ref{h1}), (\ref{h2}), (\ref{h3}), (\ref{h4}), (\ref{sigma1}) and (\ref{sigma2}), respectively.
\end{theorem} 

An important case in the empirical literature is $p=1$, that is a single endogenous variable, where we have the following result.
\begin{corollary} \label{cor:p1}
When $p=1$, under $H_0$, Assumptions \ref{assum:1} and \ref{assum:2}, as $n\rightarrow\infty$,
\begin{equation}
\sqrt{n}\left( \hat{\beta}^{2SLS}- \hat{\beta}^{OLS} \right)\stackrel{d}{\rightarrow}N(0,\sigma^2),
\end{equation}
where $ \sigma^2=\frac{1-\alpha}{\alpha}(\frac{\sigma_{u}^2}{\sigma_{vv}^2}-\frac{\sigma_{vu}^2}{\sigma_{vv}^4})$ with $\sigma_{u}^2, \sigma_{vu}$ and $\sigma_{vv}^2$ the corresponding variance and covariance of the errors.
\end{corollary}

Note that in the case of a single endogenous variable and under the null, the asymptotic variance of the difference has an explicit form. However, consistently estimating $\sigma_{u}^2$ and $\sigma_{vu}$ is infeasible with completely weak instruments. Besides, the asymptotic covariance matrix in (\ref{null_variance}) has no explicit form for multiple endogenous variables.  To implement a test procedure, we therefore propose a consistent variance/covariance estimator based on the delete-$d$ jackknife variance/covariance estimation following a theory proposed in \cite{shao1989general}. 

\subsection{Jackknife covariance estimation} \label{subsec:jack}
Let $d=\lambda n$ be an integer for $0<\lambda<1$ and $r=n-d$. Define $\mathbf{S}_{r}$ to be the collection of subsets of $\{1, \cdots, n\}$ with size $r$. For $s=\left\{i_{1}, \ldots, i_{r}\right\} \in \mathbf{S}_{r}$, let $\hat{\boldsymbol{\theta}}_s=\hat{\bbeta}_s^{2SLS}-\hat{\bbeta}_s^{OLS}$ be the subsample estimation, where $\hat{\bbeta}_s^{2SLS}$ and $\hat{\bbeta}_s^{OLS}$ are the 2SLS and OLS estimates based on the corresponding subsample. The delete-$d$ jackknife estimator of $\boldsymbol{\Sigma}_0$ is then
\begin{equation} \label{v_d}
    \widehat{\boldsymbol{\Sigma}}_0=\frac{nr}{d N} \sum_{s \in \mathbf{S}_{r}}\left(\hat{\boldsymbol{\theta}}_s-\frac{1}{N}\sum_{t \in \mathbf{S}_{r}}\hat{\boldsymbol{\theta}}_t\right)\left(\hat{\boldsymbol{\theta}}_s-\frac{1}{N}\sum_{t \in \mathbf{S}_{r}}\hat{\boldsymbol{\theta}}_t\right)',
\end{equation}
where $N={n \choose d} $. To establish the consistency and asymptotic unbiasedness of $\widehat{\boldsymbol{\Sigma}}_0$, we impose the following regularity conditions:
\begin{assumption} \label{assum:bound}
   $ \mathrm{E}(\Y'\PP_Z\Y)^{-1}=O(n^{-1}).$
\end{assumption}
This assumption is natural as $\Y'\PP_Z\Y$ has the stochastic order of $n$ under Assumptions \ref{assum:1} and \ref{assum:2}. 

\begin{theorem} \label{thm:jsve}Under $H_0$, Assumptions \ref{assum:1}, \ref{assum:2} and \ref{assum:bound}, as $n\rightarrow\infty$,
 $\widehat{\boldsymbol{\Sigma}}_0$ is a consistent and asymptotically unbiased estimator of $\boldsymbol{\Sigma}_0$.
\end{theorem} 
By using (\ref{v_d}) with large $d$, we obtain a consistent and efficient variance estimator. This is achieved at the expense of a large number of computations. As $N$ increases extremely fast as $n$ increases,  we further consider the jackknife-sampling variance estimator (JSVE) proposed in \cite{shao1989efficiency} to reduce the computation burden. Define a simple random sample (without replacement) $S_{m}$ of size $m$ from $\mathbf{S}_{r}$. We then compute $\hat{\theta}_{s}$ for $s \in S_{m}$ and apply
\begin{equation} \label{v_m}
    \widehat{\boldsymbol{\Sigma}}_0^S=\frac{nr}{d m} \sum_{s \in S_{m}}\left(\hat{\boldsymbol{\theta}}_s-\frac{1}{m}\sum_{t\in S_{m}}\hat{\boldsymbol{\theta}}_t\right)\left(\hat{\boldsymbol{\theta}}_s-\frac{1}{m}\sum_{t \in S_{m}}\hat{\boldsymbol{\theta}}_t\right)'
\end{equation}
as the variance/covariance estimator. As stated in \cite{shao1989efficiency}, JSVE is still asymptotically unbiased and consistent and we illustrate the result in the following corollary:
\begin{corollary} \label{jsve}
Under $H_0$, Assumptions \ref{assum:1}, \ref{assum:2} and \ref{assum:bound},   as $n\rightarrow\infty$, $n/m\rightarrow0$, 
 $\widehat{\boldsymbol{\Sigma}}_0^S$ is a consistent and asymptotically unbiased estimator of $\boldsymbol{\Sigma}_0$.
\end{corollary}

\begin{remark}
Note that the classical bootstrap cannot be a solution to the  weak instrument problem since it cannot replicate the correlation between instruments and structural errors in bootstrap samples \citep{anderson2010asymptotic,wang2016bootstrap}.
\end{remark}

\subsection{Test procedure} \label{subsec:test}
We now introduce our method to detect many weak instruments. The new specification testing procedure can be implemented with the following steps:
\begin{description} 
\setlength\itemsep{0.2em}
\item[Step 1:]  Obtain $N$ subsamples by randomly deleting $d$ observations from the sample without replacement for $N$ times;
\item[Step 2:] Compute $ \widehat{\boldsymbol{\Sigma}}_0^S$ based on  the $N$ subsamples taken in Step 1;
\item[Step 3:] Given $r=n-d$, compute $\boldsymbol{\theta}_s$ based on a random subsample of size $r$; and
\item[Step 4:]  Realize the test statistics $\boldsymbol{T}_n=\boldsymbol{\theta}_s'( \widehat{\boldsymbol{\Sigma}}_0^S)^{-1}\boldsymbol{\theta}_s$. Reject $H_0$ if $\boldsymbol{T}_n>\chi^2_c(p)$ at the significant level c.
\end{description}
\begin{remark}
    Step 3 ensures that the JSVE matches with the theoretical asymptotic variance. As for the choice of $\lambda$, \cite{wu1990asymptotic} recommended the interval $[0.25,0.75]$. In our simulation results, we have taken $\lambda=0.45$ which is about the middle of this interval, regardless of the number of endogenous variables considered in various settings.
\end{remark}

\section{Monte Carlo simulations} \label{sec:mc}
\subsection{Monte Carlo designs}
The goal of this section is to evaluate the finite sample performance of the test procedure introduced in the previous section. We consider the following setup with 
\begin{equation}
    y_i=\Y_i'\bbeta+u_i,
\end{equation}
\begin{equation}
    \mathbf{Y}_i=\PPi'\Z_i+\V_i,
\end{equation}
where $\bbeta=\mathbf{1}$, $\boldsymbol{Z}_i \stackrel{i.i.d.}{\sim} N_{K_n}(\mathbf{0},\mathbf{I}_{K_n})$, for $i=1,\dots, n$. The number of instruments $K_n$ varies in $\{50,100, 200\}$ with corresponding sample size $n$ satisfying three different ratios $K_n/n=\frac{1-\lambda}{3},\frac{1-\lambda}{2}$ and $\frac{2(1-\lambda)}{3}$. $\lambda$ is taken to be 0.45 as mentioned in the previous section. To investigate the effects of the number of endogenous variables, we choose $p=1,2,$ and 3. For the error terms, we generate them following two distributions: (i) multivariate normal, $N_{p+1}(\mathbf{0},\boldsymbol{\Sigma})$; and (ii) multivariate t, $t_5(\mathbf{0},\boldsymbol{\Sigma})$, where we set
\begin{equation*}
   \boldsymbol{\Sigma}=\left(\begin{array}{cc}
1 & \rho \\
\rho & 1
\end{array}\right) \;, \left(\begin{array}{ccc}
1 & \rho & \rho\\
\rho & 2 &  3 \\
\rho  & 3 &  6 
\end{array}\right) \; \text{and}\; \left(\begin{array}{cccc}
1 & \rho & \rho & \rho \\
\rho & 2 & 3 & 4\\
\rho & 3 & 6 & 10\\
\rho & 4 & 10 & 20
\end{array}\right)
\end{equation*}
for $p=1$, $p=2$ and $p=3$, respectively. The degree of endogeneity, $\rho$, is $0.9$ (highly endogenous) or $0.5$ (mildly endogenous). 

In this setup, the parameter $\PPi$ measures the strength of the instruments. To generate data under the null, we adopt the local to zero asymptotics as follows: 
\begin{equation} \label{local2zero}
    \PPi=c_n\frac{\boldsymbol{C}}{n^{1/2}},
\end{equation} with $\boldsymbol{C}=\{c_{ij}\}_{1\leq i\leq K_n, 1\leq j\leq p }\stackrel{i.i.d.}{\sim}N(0,1)$. It can be verified that under (\ref{local2zero}), $s_n$ is at the same order of the magnitude of $pc_n^2K_n$. We conclude that
\begin{itemize}
    \item $H_0\Leftrightarrow c_n=o(1)  $;
    \item $H_1  \Leftrightarrow c_n=O(1)$,
\end{itemize}
since $K_n$ and $n$ are of the same order. Thus, we set $c_n=0.1$ and $c_n=1$ to examine the sizes and powers of the proposed test, respectively. For the JSVE, we choose $m=n^{3/2}$. The Monte Carlo experiments are conducted using 1000 replications. 

\subsection{Simulation results}
Table \ref{tab:jsve} reports the empirical biases, and the root mean squared errors (RMSE) of the JSVE for each combination of $(K_n,n)$. We can see that the JSVE performs very well in these experiments.

Table \ref{tab:size} reports the empirical sizes with 5\% and 10\% nominal sizes when $c_n=0.1$. Though for $p=3$ with small ratio $K_n/n$, the proposed test is slightly oversized due to the resampling error. For example, the empirical sizes are 9.3\% and 15.8\% when $K_n=200$ and errors are Gaussian. The test still successfully controls the sizes under almost all settings,  irrespective of endogenous degrees, error distributions, combinations of $(K_n,n)$ and the number of endogenous variables. When there is a single endogenous variable (the leading empirical case), the test performs stably.

Table \ref{tab:power2} demonstrates the empirical powers of the proposed test when $c_n=1$. It has generally satisfactory power performance across the board, especially for normal errors. For a fixed ratio $K_n/n$, the empirical power increases with $K_n$  as expected. For a fixed $p$, the empirical power decreases as the ratio $K_n/n$ increases. For example, the test has powers of 85.8, 68.7 and 41.2 for three $K_n/n$ ratios when $\rho=0.9$, $p=1$, $K_n=50$ under normal errors at 95\% significant level. This is because that $\boldsymbol{\Delta}$ in (\ref{h1b}) is generally negative in our simulation settings and increases in $\alpha$. When keeping all the settings fixed but the endogenous degree $\rho$, we observe that the test has less power with mildly endogenous variables ($\rho=0.5)$ compared to the same setting with highly endogenous variables ($\rho=0.9)$. For instance, the powers are 20\% and 96.9\% when $p=2$, $K_n=50$ and $K_n/n=(1-\lambda)/3$ under normal errors at 95\% significant level, for $\rho=0.5$ and $\rho=0.9$, respectively. The reason lies in the fact that $\boldsymbol{\Delta}$ is  generally decreasing in $\rho$, according to (\ref{h1b}). When the errors follow the multivariate-t distribution, the test has less power. Especially, the test has quite low power when $\rho=0.5$ and $p=1$ under student-t errors. However, with increasing concentration parameter, say $c_n=2$,  the test shows greater power as expected even though the errors follow multivariate-t distribution. 

\begin{table} 
\renewcommand\arraystretch{0.8}
\centering
\resizebox{\columnwidth}{!}{
\begin{tabular}{cccccccccccccccc}
\hline
&\multicolumn{4}{c}{$\frac{K_n}{n}=\frac{1-\lambda}{3}$}&\multicolumn{4}{c}{$\frac{K_n}{n}=\frac{1-\lambda}{2}$}&\multicolumn{4}{c}{$\frac{K_n}{n}=\frac{2(1-\lambda)}{3}$}\\
\cmidrule(r){2-5} \cmidrule(r){6-9} \cmidrule(r){10-13} 
 $K_n$& 30 & 50 & 100 & 200 &30 & 50 & 100 & 200 &30 & 50 & 100 & 200 \\
\hline   
Bias & 0.0851 & 0.0627 &0.0458 &0.0406 & 0.0481 &0.0362 &0.0252 & 0.0179 & 0.0291  &0.0222 &0.0148 &0.0101\\
RMSE & 0.1082 & 0.0775 &0.0561 &0.0487 & 0.0626 &0.0456 &0.032 & 0.0222 & 0.0391  &0.029 &0.0189 &0.0126\\
\hline
\end{tabular}}
\caption{\small Bias and RMSE of the JSVE under $H_0$ with $m=n^{3/2}$, $\lambda=0.45$, $\rho=0.9$ and $p=1$ \label{tab:jsve}}
\end{table}

\begin{table}
\renewcommand\arraystretch{0.6}
\caption{Empirical sizes  with $m=n^{3/2}$, $\lambda=0.45$ and $c_n=0.1$ \label{tab:size}}
\centering
  \begin{adjustbox}{width=\textwidth}
\begin{tabular}{cccccccccccccc}
   \hline
   \hline
   \multicolumn{14}{c}{$\rho=0.9$}\\
\hline
&&\multicolumn{4}{c}{$p=1$}&\multicolumn{4}{c}{$p=2$}&\multicolumn{4}{c}{$p=3$}\\
\cmidrule(r){3-6} \cmidrule(r){7-10}\cmidrule(r){11-14}
                          &   &  \multicolumn{2}{c}{Normal}  & \multicolumn{2}{c}{Student-t} & \multicolumn{2}{c}{Normal}  & \multicolumn{2}{c}{Student-t} & \multicolumn{2}{c}{Normal}  & \multicolumn{2}{c}{Student-t}   \\ 
                          \cmidrule(r){3-4} \cmidrule(r){5-6}\cmidrule(r){7-8} \cmidrule(r){9-10} \cmidrule(r){11-12}\cmidrule(r){13-14}
& $K_n$& 5\% & 10\% & 5\% & 10\% &5\% & 10\% & 5\% & 10\%  &5\% & 10\% & 5\% & 10\%   \\     
\hline
\multirow{3}{*}{$\frac{K_n}{n}=\frac{1-\lambda}{3}$} 
                          & 50  & 6.7 &12.1  & 7  & 12 & 5.6 &10.2 &6.8   & 11.5 & 6.7 & 13.2 &7.3 & 12.2\\
                          & 100   & 6.4 & 11.5  &  5.7 & 9 & 6.5 &11.8 & 6.8& 13.8 & 8.6 & 13.4 & 9.4 &14.3\\ 
                          & 200  & 6.2 & 10.3 & 5.4  & 11.3 & 7.3 &12.8 & 5.8& 11 &9.3&15.8&  8&14.3 \\
                          \hline
\multirow{3}{*}{$\frac{K_n}{n}=\frac{1-\lambda}{2}$}  
                          & 50   &5.3  &10.1  & 6.6  & 10.7 & 6.4 & 11.5 & 5.9& 11 & 6.7& 12.2& 7.9  &12.4 \\
                          & 100   & 5.3 & 9.3 & 6.9  & 11.5 & 5.5 &11.3 &4.7 & 9 & 7.6& 12.3& 8.2  &13.6 \\ 
                          & 200  &6.5  & 10.7 & 3.9 &  8.6 & 5.5 &10.1 &7 &11.9  &6.9&12.1& 7 &11.2 \\
                          \hline
\multirow{3}{*}{$\frac{K_n}{n}=\frac{2(1-\lambda)}{3}$}   
                          & 50   & 3.3 & 8.9 &  6 & 10.8 & 4.9 &9.1 & 5.8&10.5  & 5.3& 9.6& 7.5 &11.9\\
                          & 100   &6.8  &11.3  &  6.7 & 10.7 & 7.7 & 12.3&5.6 &8.6  & 6.2 &12.9&6.6  &10.2 \\ 
                          & 200    & 6.5 & 10.3 &  3.6 & 8.1 & 5.1 & 10 &5.3 & 9.8 & 5.7 & 11.2& 5.7 & 10.2\\
                          \hline
   \hline
   \multicolumn{14}{c}{$\rho=0.5$}\\
   \hline
&&\multicolumn{4}{c}{$p=1$}&\multicolumn{4}{c}{$p=2$}&\multicolumn{4}{c}{$p=3$}\\
\cmidrule(r){3-6} \cmidrule(r){7-10}\cmidrule(r){11-14}
                          &   &  \multicolumn{2}{c}{Normal}  & \multicolumn{2}{c}{Student-t} & \multicolumn{2}{c}{Normal}  & \multicolumn{2}{c}{Student-t} & \multicolumn{2}{c}{Normal}  & \multicolumn{2}{c}{Student-t}   \\ 
                          \cmidrule(r){3-4} \cmidrule(r){5-6}\cmidrule(r){7-8} \cmidrule(r){9-10} \cmidrule(r){11-12}\cmidrule(r){13-14}
& $K_n$& 5\% & 10\% & 5\% & 10\% &5\% & 10\% & 5\% & 10\%  &5\% & 10\% & 5\% & 10\%   \\     
\hline
\multirow{3}{*}{$\frac{K_n}{n}=\frac{1-\lambda}{3}$} 
                          & 50  & 5.2 & 10 & 6.1  & 11.1 & 6.1 & 11.5 & 6.3 & 10.5 &6.1&12.1& 6.6 &12.2 \\
                          & 100  & 5.8 & 10.9 & 4.8  & 9.9 & 6.1 & 12 & 7& 12.3 & 6.9&11.5&7.8  &13.1\\ 
                          & 200  &7  & 13.2 & 7.5  &  12.4 &6  &10.7 & 6.4& 10.9 &7&12.4&8.1  &13.4 \\
                          \hline
\multirow{3}{*}{$\frac{K_n}{n}=\frac{1-\lambda}{2}$} 
                          & 50   & 4.9 &9.5  & 7.4  & 12.4 & 6.2 & 11.4 & 5.7 &  10.2 & 6.5 &12.2& 6.6 &11 \\
                          & 100   & 5.1  & 9.6 & 5.9  & 10.3 & 5.4 & 10.8 & 5.2 & 8.9  & 5.9 & 10.6&7  &11.5 \\ 
                          & 200  &  4.9& 10.2   & 5.4  & 9.8  & 5.5 &10.4 & 6.8 & 11.7 &5.7&11.4& 5.8 &11 \\
                          \hline
\multirow{3}{*}{$\frac{K_n}{n}=\frac{2(1-\lambda)}{3}$}     
                          & 50   & 4.5 & 8.5 &  4.5 & 8.6 & 5 &10.9 & 5.9 & 9.9  &5.5&10.4&5.3  &10.1 \\
                          & 100   & 5.8 & 10.8 &  6.2 & 9.9 &5.1  &10.4 &6.4 &9.8  &5.5&10.8&5.4  &9.4\\ 
                          & 200   & 5.2 &9.7  & 5.5  & 9 & 4.5 &10 &5.6 &  9&3.9&8.9&5.2  &9 \\
                          \hline
   \hline
\end{tabular}
\end{adjustbox}
\end{table}

\begin{table}
\renewcommand\arraystretch{0.6}
\caption{Empirical powers  with $m=n^{3/2}$, $\lambda=0.45$, and $c_n=1$ \label{tab:power2}}
\centering
  \begin{adjustbox}{width=\textwidth}

\begin{tabular}{cccccccccccccc}
   \hline
   \hline
   \multicolumn{14}{c}{$\rho=0.9$}\\
\hline
&&\multicolumn{4}{c}{$p=1$}&\multicolumn{4}{c}{$p=2$}&\multicolumn{4}{c}{$p=3$}\\
\cmidrule(r){3-6} \cmidrule(r){7-10}\cmidrule(r){11-14}
                          &   &  \multicolumn{2}{c}{Normal}  & \multicolumn{2}{c}{Student-t} & \multicolumn{2}{c}{Normal}  & \multicolumn{2}{c}{Student-t} & \multicolumn{2}{c}{Normal}  & \multicolumn{2}{c}{Student-t}   \\ 
                          \cmidrule(r){3-4} \cmidrule(r){5-6}\cmidrule(r){7-8} \cmidrule(r){9-10} \cmidrule(r){11-12}\cmidrule(r){13-14}
& $K_n$& 5\% & 10\% & 5\% & 10\% &5\% & 10\% & 5\% & 10\%  &5\% & 10\% & 5\% & 10\%   \\     
\hline
\multirow{3}{*}{$\frac{K_n}{n}=\frac{1-\lambda}{3}$} 
                          & 50  & 85.8 & 91.8 & 14.4  &  21.4 &  96.9&  99.3 &  31.8 &42 & 99&99.8&62.3 &71.5 \\
                          & 100   &100  & 100 & 30.1  & 41.2 & 99.5 & 99.7 & 57.5  &65.6 &100& 100& 78.1 & 83 \\ 
                          & 200  &  100&  100 &37.3   &48.7  & 100 & 100 & 67.7 &75.6 & 100&100&86.9 & 90.1 \\
                          \hline
\multirow{3}{*}{$\frac{K_n}{n}=\frac{1-\lambda}{2}$} 
                          & 50  & 68.7 & 79.1 & 15.3 & 21.4 & 55.7  & 69.1 & 25.8 & 33.5 &99&99.8& 46.7 &58.3 \\
                          & 100  & 99.5 &  99.8 & 23.7 & 33.1  & 91.9 & 95.8 & 46.5 & 56.7 &100 &100& 68.2 & 76.7 \\ 
                          & 200  & 100 & 100 & 23.4  & 31.7 & 100 & 100& 65.4  &72.8 &100&100& 79.4 &84.1 \\
                          \hline
\multirow{3}{*}{$\frac{K_n}{n}=\frac{2(1-\lambda)}{3}$}   
                          & 50   & 41.2 & 52 & 14.8  & 21.5 & 30.1 & 42.6 & 21.6 & 31.3 & 39.5 & 51.2 & 25.3 & 34.5\\
                          & 100   & 94.7 & 98.1 & 14.4  & 21.9 & 71.4 & 81.7 &  30.5 &42.7 &83.1 & 89.9&47.1 &56.1 \\ 
                          & 200   &99.9  & 99.9 & 24.4  & 32.9 & 94.9 & 98.1 & 45.4  & 54.3 & 99.1 & 99.6 & 66.7 &73.2 \\
                          \hline
   \hline
   \multicolumn{14}{c}{$\rho=0.5$}\\
   \hline
&&\multicolumn{4}{c}{$p=1$}&\multicolumn{4}{c}{$p=2$}&\multicolumn{4}{c}{$p=3$}\\
\cmidrule(r){3-6} \cmidrule(r){7-10}\cmidrule(r){11-14}
                          &   &  \multicolumn{2}{c}{Normal}  & \multicolumn{2}{c}{Student-t} & \multicolumn{2}{c}{Normal}  & \multicolumn{2}{c}{Student-t} & \multicolumn{2}{c}{Normal}  & \multicolumn{2}{c}{Student-t}   \\ 
                          \cmidrule(r){3-4} \cmidrule(r){5-6}\cmidrule(r){7-8} \cmidrule(r){9-10} \cmidrule(r){11-12}\cmidrule(r){13-14}
& $K_n$& 5\% & 10\% & 5\% & 10\% &5\% & 10\% & 5\% & 10\%  &5\% & 10\% & 5\% & 10\%   \\     
\hline
\multirow{3}{*}{$\frac{K_n}{n}=\frac{1-\lambda}{3}$}  
                          & 50   & 20.9  & 30.2  &  7.6 & 13.2 & 20 &31  &14 & 21.8 & 51.8& 63& 21  &30 \\
                          & 100    &  52.4 & 64.1  &  8.3 & 14.4 & 55.4 &65.9 &17.7 & 26.7 & 85.3 & 90.6&27.7  &37.9\\ 
                          & 200   & 79.5 & 86.8 & 12.1  & 19.5 & 88.1 & 93.7 &  25.7 &35.8  &98.2&99.1&39.2  &49.6 \\
                          \hline
\multirow{3}{*}{$\frac{K_n}{n}=\frac{1-\lambda}{2}$}  
                          & 50   & 16.8 & 23.8 & 8.7  & 13.4 & 17.9 & 27.8 &10.8 & 16.4 & 28.8& 37.6&17.2  &24.5 \\
                          & 100   & 38.1  & 50.3 &  9.3 & 14 & 35.4 &46.2 &17.2 & 25.9 & 53.8&66.2&23.3  &32.4\\ 
                          & 200   & 66.3 & 76.5 & 7.4  & 13.6 & 63.3 & 73.8 &21.8 & 30.8 & 85.6& 91.3&27.2  &36.2 \\
                          \hline
\multirow{3}{*}{$\frac{K_n}{n}=\frac{2(1-\lambda)}{3}$}   
                          & 50   & 9.6 & 15.2 & 6  & 9.7 & 13.9 &20.9 &9 &14.3  &13.2&20&9.2  &15.3 \\
                          & 100   & 27.3  & 37.3 & 7.2  & 10.7 & 19.3 & 28.8&11.5 & 19.3 &26.3&36.1&14.3  &23.3\\ 
                          & 200   & 36.9 & 48.3 & 7.1  &13  & 39.1 &53 & 14&19.7  &49.8&61.4&20.4  &30.1 \\
                          \hline
   \hline
\end{tabular}
\end{adjustbox}
\end{table}

\section{An empirical illustration: Return to education} 
\label{sec:emp}
In this section, we re-analyse the returns to education data of \cite{angrist1991does} (henceforth referred to as AK1991) using quarter of birth as an instrument for educational attainment. One of the setups in the original AK1991 application use up to 180 instruments that include 30 quarter and year of birth interactions and 150 quarter and state of birth interactions. Later it has been widely suggested that the setup suffers from a weak instrument problem (\cite{angrist1995split}; \cite{bound1995problems}). MS2022 appled their proposed $\widetilde{F}$ test and  argued that the instrument set is not completely weak with the original full data. 

As the original sample size (329,509) of the data is larger than usual for empirical research, we consider subsamples with 0.1\% ($n=330$) and 0.5\% ($n=1650$) of the original sample size, more in keeping with the typical empirical application. We examine the specification of 180 instruments and 1350 instruments, respectively, that extend the model by including the interactions among quarter and year and state of birth. We evaluate the performance of the OLS estimator, the 2SLS estimator,  the $\widetilde{F}$ test and our proposed test based on 1000 randomly chosen subsamples and report the results in Table \ref{tab:emp}. Within our specification with 0.1\% subsamples ($n=330,K_n=180)$, we find that the OLS and 2SLS estimates are close in magnitude: the average of OLS is 0.0483 and the average of 2SLS is 0.0514. Result in column (d) is in favor of this finding where our proposed test shows that the two estimates are close significantly. The average $\widetilde{F}$ is 4.84, which is above the cut-off 4.14 suggested by MS2022. It provides evidence that 1\%-scheme produces subsamples which are not completely weak.  Therefore, our test helps provide more information and explain why the OLS and 2SLS estimators behave similarly: The instruments in 1\% subsample are moderately weak, that is, $s_n/\sqrt{n} \rightarrow \infty$ but $s_n/n \rightarrow 0$.

For the 0.5\% subsamples with $K_n=1350$ instruments and $n=1650$, the difference between OLS and 2SLS estimators is again negligible as shown in columns (a), (b) and (d). It therefore indicates the presence of weak instruments. The $\widetilde{F}$ test then can provide detailed categorization: it average value is 4.65, evidencing that the instruments are moderately weak. Though not as precise as the result obtained from the $\widetilde{F}$ test, our test is still informative to identify the strength of instruments.

\begin{table}[]\renewcommand\arraystretch{0.6}
    \centering
    \begin{tabular}{ccccccccc}
    \hline
    \hline
           $n$ & $K_n$ & (a) & (b) & (c) &  (d) & (e) & (f) \\
           \hline
         330 & 180 & 0.0483 & 0.0514 & 4.84 & 0.0010& 57.8\% & 0 \\
        1650 & 1350 & 0.0636 & 0.0618 & 4.65& -0.0004 & 62.1\% & 0\\
         \hline
         \hline
    \end{tabular}
    \caption{Empirical Results  \label{tab:emp}}
    \begin{tablenotes}
     \item (a): average value of the OLS estimators; (b): average value of the 2SLS estimators; (c): average value of the $\widetilde{F}$ statistics; (d): average value of our proposed test; (e): rejection rate of 
     $s_n/\sqrt{n}\rightarrow \kappa_2\in [0,\infty)$ by $\widetilde{F}$; (f): rejection rate of  $s_n/n\rightarrow 0$ by our proposed test.
    \end{tablenotes}
\end{table}

\section{Conclusion}
\label{sec:con}
We find that the limits of the 2SLS and OLS estimators coincide under the many weak instruments asymptotics but differ under the many strong instruments asymptotics. Building on this, we propose a test statistic to distinguish between these two specifications.  The proposed test allows for multiple endogenous variables and shows robustness to general forms of non-normality in the error distribution. 

Our proposed test can be applied to determine which asymptotic scheme should be used. We suggest that the many strong asymptotics are trustworthy if the 2SLS and OLS estimates are not significantly close. Estimation methods can be applied while achieving the optimal convergence rate, and so as inferential tools leaning on the assumption of many strong instruments. 

Besides, our proposed test can be viewed as a follow-up test for the $\widetilde{F}$ test in MS2022. If empirical researchers draw conclusion that the large instrument set is not completely weak, then one can apply our proposed test to differentiate the case with strong instruments from the case with moderately weak instruments.

\appendix

\section{Proofs of main results}\label{app}


\subsection{Proof of Theorem  \ref{theo:limit}}
Before proving Theorem  \ref{theo:limit}, we first state a lemma that is used in the proofs.
\begin{lemma}
Under Assumption \ref{assum:1} and \ref{assum:2}, as $n\rightarrow \infty$, the following statements are true: 
\begin{itemize} \label{lem:chao}
    \item[(a)] $\frac{\PPi'\Z'\uu}{K_n}\stackrel{p}{\rightarrow} \boldsymbol{0}_{p\times1}$;
    \item[(b)] $\frac{\PPi'\Z'\V}{K_n}\stackrel{p}{\rightarrow} \boldsymbol{0}_{p\times p}$;
    \item[(c)] $\frac{\V'\PP_Z\uu}{K_n}\stackrel{p}{\rightarrow} \boldsymbol{\Sigma}_{Vu}$;
    \item[(d)] $\frac{\V'\PP_Z\V}{K_n}\stackrel{p}{\rightarrow} \boldsymbol{\Sigma}_{VV}$;
    \item[(e)] $\frac{\V'\uu}{n}\stackrel{p}{\rightarrow} \boldsymbol{\Sigma}_{Vu}$;
    \item[(f)] $\frac{\V'\V}{n}\stackrel{p}{\rightarrow} \boldsymbol{\Sigma}_{VV}$;
\end{itemize}
\end{lemma} \label{lem:limit}
\begin{proof}
(a), (b), (c), (d) are directly obtained from Lemma A1 established in \cite{chao2005consistent}. To prove part (e), note that it suffices to prove that $\V_{(j)}^{\prime} \uu / n\stackrel{p}{\rightarrow}\boldsymbol{\Sigma}_{V u}^{(j)}$ as $n \rightarrow \infty$, so that $\V_{(j)}^{\prime} \uu / n$ is the $j$-th element of $\V'\uu/n$, and where $\boldsymbol{\Sigma}_{V u}^{(j)}$ denotes the $j$-th element of $\boldsymbol{\Sigma}_{V u}$. In fact, by Assumption \ref{assum:2} and law of large number, it holds naturally. Part (f) follow from proof similar to that of part (e).
\end{proof}

\begin{proof}[Proof of Theorem \ref{theo:limit}]
It can be verified that
\begin{equation}
    \hat{\bbeta}^{OLS}-\bbeta=(\Y'\Y)^{-1}\Y'\uu,  \nonumber
\end{equation}
\begin{equation}
    \hat{\bbeta}^{2SLS}-\bbeta=(\Y'\PP_{Z}\Y)^{-1}\Y'\PP_{Z}\uu. \nonumber
\end{equation}
Further,
\begin{equation} \label{comp1}
    \Y'\PP_{Z}\Y=\PPi'\Z'\Z\PPi+\PPi'\Z'\V+\V'\Z\PPi+\V'\PP_{Z}\V, 
    \end{equation}
\begin{equation}\label{comp2}
    \Y'\PP_{Z}\uu=\PPi'\Z'\uu+\V'\PP_{Z}\uu,
\end{equation}
\begin{equation}\label{comp3}
    \Y'\Y=\PPi'\Z'\Z\PPi+\PPi'\Z'\V+\V'\Z\PPi+\V'\V,
\end{equation}
and
\begin{equation}\label{comp4}
    \Y'\uu=\PPi'\Z'\uu+\V'\uu.
\end{equation}
Hence, we can write 
\begin{equation}
    \frac{\Y'\PP_{Z}\Y}{n}=\frac{s_n}{n}\frac{\PPi'\Z'\Z\PPi}{s_n}+\frac{K_n}{n}\left(\frac{\PPi'\Z'\V}{K_n}+\frac{\V'\Z\PPi}{K_n}+\frac{\V'\PP_{Z}\V}{K_n}\right),  \nonumber
\end{equation}
\begin{equation}
    \frac{\Y'\PP_{Z}\uu}{n}=\frac{K_n}{n}\left(\frac{\PPi'\Z\uu}{K_n}+\frac{\V'\PP_{Z}\uu}{K_n} \nonumber\right),
\end{equation}
\begin{equation}
    \frac{\Y'\Y}{n}=\frac{s_n}{n}\frac{\PPi'\Z'\Z\PPi}{s_n}+\frac{K_n}{n}\left(\frac{\PPi'\Z'\V}{K_n}+\frac{\V'\Z\PPi}{K_n}\right)+\frac{\V'\V}{n},  \nonumber
\end{equation}
and
\begin{equation}
    \frac{\Y'\uu}{n}=\frac{\PPi'\Z\uu}{n}+\frac{\V'\uu}{n}.  \nonumber
\end{equation}
To show part (a), under $H_1$, by Assumption \ref{assum:1} and parts (b) and (d) in Lemma \ref{lem:chao}, we have 
\begin{equation}
      \frac{\Y'\PP_{Z}\Y}{n} \stackrel{p}{\rightarrow} \kappa_0\boldsymbol{\Theta}+\alpha\boldsymbol{\Sigma}_{VV},  \nonumber
\end{equation}
Similarly, 
\begin{equation}
      \frac{\Y'\PP_{Z}\uu}{n} \stackrel{p}{\rightarrow}  \alpha\boldsymbol{\Sigma}_{Vu},  \nonumber
\end{equation}
\begin{equation}
      \frac{\Y'\Y}{n} \stackrel{p}{\rightarrow} \kappa_0\boldsymbol{\Theta}+\boldsymbol{\Sigma}_{VV},  \nonumber
\end{equation}
and
\begin{equation}  \nonumber
      \frac{\Y'\uu}{n} \stackrel{p}{\rightarrow} \boldsymbol{\Sigma}_{Vu}.
\end{equation}
Thus, under $H_1$, it follows immediately from the Slutsky theorem that 
\begin{equation}
    \hat{\bbeta}^{2SLS}\stackrel{p}{\rightarrow} \bbeta + (\kappa_0/\alpha\boldsymbol{\Theta}+\boldsymbol{\Sigma}_{VV})^{-1}\boldsymbol{\Sigma}_{Vu},  \nonumber
\end{equation}
and
\begin{equation}
    \hat{\bbeta}^{OLS}\stackrel{p}{\rightarrow} \bbeta + (\kappa_0\boldsymbol{\Theta}+\boldsymbol{\Sigma}_{VV})^{-1}\boldsymbol{\Sigma}_{Vu}.  \nonumber
\end{equation} Note that $s_n/n \rightarrow 0$ under $H_0$, the proof of part (b) is similar, for brevity, we omit its proof here. 
\end{proof}

\subsection{Proof of Theorem  \ref{theo:stat}}
In this subsection, we first introduce Lemma \ref{lem:wang}, which establishes the joint CLT of sesquilinear form, then followed by Lemma \ref{lem:7comp} applying Lemma \ref{lem:wang}. 
\begin{lemma}(Theorem 2.1, \cite{wang2014joint}) \label{lem:wang}
Let $\left\{\mathbf{A}_{n}\right\}$ and $\left\{\mathbf{B}_{n}\right\}$ be two sequences of $n \times n$ symmetric matrices. Assume that the following limits exist:
$$w_{1}=\lim _{n \rightarrow \infty} \frac{1}{n} \operatorname{tr}\left[\mathbf{A}_{n} \circ \mathbf{A}_{n}\right], \; w_{2}=\lim _{n \rightarrow \infty} \frac{1}{n} \operatorname{tr}\left[\mathbf{B}_{n} \circ \mathbf{B}_{n}\right],\; w_{3}=\lim _{n \rightarrow \infty} \frac{1}{n} \operatorname{tr}\left[\mathbf{A}_{n} \circ \mathbf{B}_{n}\right]$$
$$\theta_{1}=\lim _{n \rightarrow \infty} \frac{1}{n} \operatorname{tr}\left[\mathbf{A}_{n} \mathbf{A}_{n}'\right],  \;\theta_{2}=\lim _{n \rightarrow \infty} \frac{1}{n} \operatorname{tr}\left[\mathbf{B}_{n} \mathbf{B}_{n}'\right],\; \theta_{3}=\lim _{n \rightarrow \infty} \frac{1}{n} \operatorname{tr}\left[\mathbf{A}_{n} \mathbf{B}_{n}'\right],$$
where $\mathbf{A} \circ \mathbf{B}$ denotes the Hadamard product of two matrices $\mathbf{A}$ and $\mathbf{B}$, i.e. $(\mathbf{A} \circ \mathbf{B})_{i j}=$ $\mathbf{A}_{i j} \cdot \mathbf{B}_{i j}$. For matrices $\widetilde{\mathbf{X}}$ and $\widetilde{\mathbf{Y}}$ $\in \mathbb{R}^{n \times m}$, define two groups of sesquilinear forms:
$$
\tilde{U}(l)=\frac{1}{\sqrt{n}}\bigg(\widetilde{\mathbf{X}}(l)' \mathbf{A}_{n} \widetilde{\mathbf{Y}}(l)-\mathrm{E}(\widetilde{\mathbf{X}}(l)' \mathbf{A}_{n} \widetilde{\mathbf{Y}}(l)) \bigg),$$ $$ \tilde{V}(l)=\frac{1}{\sqrt{n}}\bigg(\widetilde{\mathbf{X}}(l)' \mathbf{B}_{n} \widetilde{\mathbf{Y}}(l)-\mathrm{E}(\widetilde{\mathbf{X}}(l)' \mathbf{B}_{n} \widetilde{\mathbf{Y}}(l))\bigg).
$$
Then, the $2m$-dimensional random vector:
$$
(\tilde{U}(1), \cdots, \tilde{U}(m), \tilde{V}(1), \cdots, \tilde{V}(m))'
$$
converges weakly to a zero-mean Gaussian vector $\boldsymbol{\omega}$. Moreover, the covariance matrix of $\boldsymbol{\omega}$ is
$$
\boldsymbol{\Lambda}=\left(\begin{array}{ll}
\boldsymbol{\Lambda}_{11} & \boldsymbol{\Lambda}_{12} \\
\boldsymbol{\Lambda}_{12} & \boldsymbol{\Lambda}_{22}
\end{array}\right)_{2 m\times 2 m}.
$$
Each block within $\boldsymbol{\Lambda}$ is a $m \times m$ matrix, having the following structure:
$$
\begin{array}{l}
\boldsymbol{\Lambda}_{11,ij}=w_{1} \mathrm{Cov}(\tilde{x}_{1i}\tilde{y}_{1i},\tilde{x}_{1j}\tilde{y}_{1j})+\left(\theta_{1}-w_{1}\right) \mathrm{E}(\tilde{x}_{1i}\tilde{x}_{1j})\mathrm{E}(\tilde{y}_{1i}\tilde{y}_{1j})+\left(\theta_{1}-w_{1}\right) \mathrm{E}(\tilde{x}_{1i}\tilde{y}_{1j})\mathrm{E}(\tilde{x}_{1j}\tilde{y}_{1i}), \\
\boldsymbol{\Lambda}_{22,ij}=w_{2} \mathrm{Cov}(\tilde{x}_{1i}\tilde{y}_{1i},\tilde{x}_{1j}\tilde{y}_{1j})+\left(\theta_{2}-w_{2}\right) \mathrm{E}(\tilde{x}_{1i}\tilde{x}_{1j})\mathrm{E}(\tilde{y}_{1i}\tilde{y}_{1j})+\left(\theta_{2}-w_{2}\right) \mathrm{E}(\tilde{x}_{1i}\tilde{y}_{1j})\mathrm{E}(\tilde{x}_{1j}\tilde{y}_{1i}),\\
\boldsymbol{\Lambda}_{12,ij}=w_{3} \mathrm{Cov}(\tilde{x}_{1i}\tilde{y}_{1i},\tilde{x}_{1j}\tilde{y}_{1j})+\left(\theta_{3}-w_{3}\right) \mathrm{E}(\tilde{x}_{1i}\tilde{x}_{1j})\mathrm{E}(\tilde{y}_{1i}\tilde{y}_{1j})+\left(\theta_{3}-w_{3}\right) \mathrm{E}(\tilde{x}_{1i}\tilde{y}_{1j})\mathrm{E}(\tilde{x}_{1j}\tilde{y}_{1i}).
\end{array}
$$

\end{lemma}
Before showing the proof of Theorem \ref{theo:stat}, we firstly define the following variables:
\begin{equation*}
    \widetilde{\Z}:=\Z\PPi=\left[\widetilde{\Z}(1),\cdots,\widetilde{\Z}(p)\right]\in \mathbb{R}^{n \times p}, \; \;    \mathbf{U}:=\left[\uu,\cdots,\uu\right] \in \mathbb{R}^{n \times p},
\end{equation*}
\begin{equation*}
    \Dot{\Z}:=\left[\underbrace{\overbrace{\widetilde{\Z}(1),\cdots,\widetilde{\Z}(1)}^{p}\\,\cdots,\overbrace{\widetilde{\Z}(p),\cdots,\widetilde{\Z}(p)}^p}_{p^2}\right] \in \mathbb{R}^{n \times p^2},
\end{equation*}
 \begin{equation*}
     \Ddot{\Z}:=\left[\underbrace {\overbrace{\widetilde{\Z}(1),\cdots,\widetilde{\Z}(p)}^{p}\\,\cdots,\overbrace{\widetilde{\Z}(1),\cdots,\widetilde{\Z}(p)}^p}_{p^2}\right] \in \mathbb{R}^{n \times p^2},
 \end{equation*} 
\begin{equation*}
    \Dot{\V}:=\left[\underbrace{\overbrace{\V(1),\cdots,\V(1)}^{p},\cdots,\overbrace{\V(p),\cdots,\V(p)}^p}_{p^2}\right] \in \mathbb{R}^{n \times p^2},
\end{equation*}
and
\begin{equation*}
    \Ddot{\V}:=\left[\underbrace{\overbrace{\V(1),\cdots,\V(p)}^{p},\cdots,\overbrace{\V(1),\cdots,\V(p)}^p}_{p^2}\right] \in \mathbb{R}^{n \times p^2}.
\end{equation*} Let
\begin{equation}
    \widetilde{\X}:=\left[\Dot{\Z},\Dot{\Z},\widetilde{\Z},\Dot{\V},\V\right]=\left[\widetilde{\X}(1),\cdots,\widetilde{\X}(3p^2+2p) \right] \in \mathbb{R}^{n \times (3p^2+2p)} \nonumber
\end{equation}
and
\begin{equation}
    \widetilde{\Y}:=\left[\Ddot{\Z},\Ddot{\V},\mathbf{U},\Ddot{\V},\mathbf{U}\right]=\left[\widetilde{\Y}(1),\cdots,\widetilde{\Y}(3p^2+2p) \right] \in \mathbb{R}^{n \times (3p^2+2p)}.  \nonumber
\end{equation} Further we introduce the following quantities: 
\begin{equation} \label{eq:sigma11}
    \boldsymbol{\Sigma}_{11,ij}=\mathrm{Cov}(\Tilde{x}_{1i}\Tilde{y}_{1i},\Tilde{x}_{1j}\Tilde{y}_{1j}), 
\end{equation}
\begin{equation} \label{eq:sigma22}
    \boldsymbol{\Sigma}_{22,ij}=\alpha^2 \boldsymbol{\Sigma}_{11,ij}+(\alpha-\alpha^2)\left(\mathrm{E}(\Tilde{x}_{1i}\Tilde{x}_{1j})\mathrm{E}(\Tilde{y}_{1i}\Tilde{y}_{1j})+\mathrm{E}(\Tilde{x}_{1i}\Tilde{y}_{1j})\mathrm{E}(\Tilde{x}_{1j}\Tilde{y}_{1i})\right). 
\end{equation}

We then establish the joint distribution of seven key components in the following lemma:
\begin{lemma} \label{lem:7comp}
Under Assumption \ref{assum:1} and \ref{assum:2}, as $n\rightarrow \infty$, we have
$$
\sqrt{n}\cdot\boldsymbol{\Sigma}_0^{-\frac{1}{2}}\cdot\left(\left(\begin{array}{l}
n^{-1} \mathrm{vec}(\PPi'\Z'\Z\PPi) \\
n^{-1} \mathrm{vec}(\PPi'\Z'\V) \\
n^{-1} \PPi'\Z'\uu\\
n^{-1} \mathrm{vec}(\V'\V) \\
n^{-1} \V'\uu \\
n^{-1} \mathrm{vec}(\V'\PP_Z\V) \\
n^{-1}\V'\PP_Z\uu
\end{array}\right)-\left(\begin{array}{c}
n^{-1} \mathrm{E}(\mathrm{vec}(\PPi'\Z'\Z\PPi)) \\
\mathbf{0} \\
\mathbf{0} \\
\mathrm{vec}(\boldsymbol{\Sigma}_{VV})\\
\boldsymbol{\Sigma}_{uV}\\
\frac{K_n}{n} \boldsymbol{\Sigma}_{VV}\\
\frac{K_n}{n} \boldsymbol{\Sigma}_{uV}
\end{array}\right)\right) \stackrel{d}{\rightarrow} N\left(\mathbf{0}, \mathbf{I}_{4p^2+3p}\right),
$$
where $$\boldsymbol{\Sigma}_0=\left[\begin{array}{cccc}
\boldsymbol{\Sigma}_1 &  & & \\
 & \boldsymbol{\Sigma}_2 & & \\
&& \boldsymbol{\Sigma}_3 & \alpha \boldsymbol{\Sigma}_3 \\
& &\alpha \boldsymbol{\Sigma}_3  & \boldsymbol{\Sigma}_4
\end{array}\right]$$ with
\begin{equation} \label{sigma1}
    \boldsymbol{\Sigma}_{1}=\{\boldsymbol{\Sigma}_{11,ij}\}_{1\leq i\leq p^2, 1\leq j\leq p^2}\in \mathbb{R}^{p^2\times p^2},
\end{equation}
\begin{equation} \label{sigma2}
    \boldsymbol{\Sigma}_{2}=\{\boldsymbol{\Sigma}_{11,ij}\}_{p^2+1\leq i\leq 2p^2+p , p^2+1\leq j\leq 2p^2+p}\in \mathbb{R}^{(p^2+p)\times (p^2+p)},
\end{equation}
\begin{equation}\label{sigma3}
    \boldsymbol{\Sigma}_{3}=\{\boldsymbol{\Sigma}_{11,ij}\}_{2p^2+p+1\leq i\leq 3p^2+2p , 2p^2+p+1\leq j\leq 3p^2+2p}\in \mathbb{R}^{(p^2+p)\times (p^2+p)},
\end{equation}
and  
\begin{equation}\label{sigma4}
    \boldsymbol{\Sigma}_{4}=\{\boldsymbol{\Sigma}_{22,ij}\}_{2p^2+p+1\leq i\leq 3p^2+2p , 2p^2+p+1\leq j\leq 3p^2+2p}\in \mathbb{R}^{(p^2+p)\times (p^2+p)}.
\end{equation}
\end{lemma}

\begin{proof}
We then apply Lemma \ref{lem:wang} by setting $\mathbf{A}_n=\mathbf{I}_n$ and $\mathbf{B}_n=\mathbf{P}_Z$. We verify that $\omega_1=\theta_1=1$, $\theta_2=\omega_3=\theta_3=\alpha$ and $\omega_2$ exists since the leverage value $P_{ii}$ ranges in $[0,1]$. Note that $\omega_2=\alpha^2+o(1)$ by Assumption \ref{assum:2}(a) \citep{anatolyev2017asymptotics}. For $1\leq l \leq 3p^2+2p$, define two groups of sesquilinear forms as following,
\begin{equation*}
    \Tilde{U}(l):=\frac{1}{\sqrt{n}}\left(\Tilde{\X}(l)'\Tilde{\Y}(l)-\mathrm{E}(\Tilde{\X}(l)'\Tilde{\Y}(l)) \right)=\sqrt{n}\left(\Tilde{\X}(l)'\Tilde{\Y}(l)/n-\mathrm{E}(\Tilde{\X}(l)'\Tilde{\Y}(l)/n )\right),
\end{equation*}
\begin{equation*}
    \Tilde{V}(l):=\frac{1}{\sqrt{n}}\left(\Tilde{\X}(l)'\PP_{Z}\Tilde{\Y}(l)-\mathrm{E}(\Tilde{\X}(l)'\PP_{Z}\Tilde{\Y}(l)) \right)=\sqrt{n}\left(\Tilde{\X}(l)'\PP_{Z}\Tilde{\Y}(l)/n-\mathrm{E}(\Tilde{\X}(l)'\PP_{Z}\Tilde{\Y}(l)/n )\right).
\end{equation*}
It holds that the $(6p^2+4p)$-dimensional random vector
\begin{equation} \label{ses}
    \begin{bmatrix}
  \boldsymbol{\Sigma}_{11} & \boldsymbol{\Sigma}_{12} \\ \boldsymbol{\Sigma}_{21} & \boldsymbol{\Sigma}_{22}
\end{bmatrix}^{-\frac{1}{2}} \cdot \left(\Tilde{U}(1),\cdots,\Tilde{U}(3p^2+2p),\Tilde{V}(1),\cdots,\Tilde{V}(3p^2+2p)\right)'\stackrel{d}{\rightarrow}N\left(\mathbf{0}, \mathbf{I}_{6p^2+4p}\right).
\end{equation} where each block of the covariance matrix is a $(3p^2+2p)\times(3p^2+2p)$ matrix with $\boldsymbol{\Sigma}_{11}$ and $\boldsymbol{\Sigma}_{22}$ given in (\ref{eq:sigma11}) and (\ref{eq:sigma22}), and $\boldsymbol{\Sigma}_{12}=\alpha\boldsymbol{\Sigma}_{11}$.
Further, by Assumption \ref{assum:1},  it can be shown that $\boldsymbol{\Sigma}_{11}=\mathrm{diag}(\boldsymbol{\Sigma}_{1},\boldsymbol{\Sigma}_{2},\boldsymbol{\Sigma}_{3})$ and $\boldsymbol{\Sigma}_{22}=\mathrm{diag}(\boldsymbol{\Lambda}_{1},\boldsymbol{\Lambda}_{2},\boldsymbol{\Sigma}_{4})$ 
with $\boldsymbol{\Sigma}_{1}=\{\boldsymbol{\Sigma}_{11,ij}\}_{1\leq i\leq p^2, 1\leq j\leq p^2} $, $\boldsymbol{\Sigma}_{2}=\{\boldsymbol{\Sigma}_{11,ij}\}_{p^2+1\leq i\leq 2p^2+p , p^2+1\leq j\leq 2p^2+p}$, $\boldsymbol{\Sigma}_{3}=\{\boldsymbol{\Sigma}_{11,ij}\}_{2p^2+p+1\leq i\leq 3p^2+2p , 2p^2+p+1\leq j\leq 3p^2+2p} $, $\boldsymbol{\Lambda}_{1}=\{\boldsymbol{\Sigma}_{22,ij}\}_{1\leq i\leq p^2, 1\leq j\leq p^2} $, $\boldsymbol{\Lambda}_{2}=\{\boldsymbol{\Sigma}_{22,ij}\}_{p^2+1\leq i\leq 2p^2+p , p^2+1\leq j\leq 2p^2+p}$, \\and  $\boldsymbol{\Sigma}_{4}=\{\boldsymbol{\Sigma}_{22,ij}\}_{2p^2+p+1\leq i\leq 3p^2+2p , 2p^2+p+1\leq j\leq 3p^2+2p}$. Lemma \ref{lem:7comp} then holds as a corollary of (\ref{ses}) by filtering out the $\left(\Tilde{V}(1),\cdots,\Tilde{V}(2p^2+2p)\right)$.
\end{proof}

\begin{proof}[Proof of Theorem \ref{theo:stat}]
We denote the CLT of the seven key components in Lemma \ref{lem:7comp} as $\sqrt{n}(\boldsymbol{\theta}_n-\boldsymbol{\theta})\stackrel{d}{\rightarrow}\boldsymbol{\xi}$. Note that $\hat{\bbeta}^{OLS}$ and $\hat{\bbeta}^{2SLS}$ are linear combinations of the seven key components, we apply the delta method to Lemma \ref{lem:7comp} to establish the limiting distribution of $\hat{\bbeta}^{2SLS}-\hat{\bbeta}^{OLS}$.
Define $\boldsymbol{f}:\mathbb{R}^{4p^2+3p}\rightarrow\mathbb{R}^{p}$ s.t.
\begin{align*}
     \boldsymbol{f}\left( \boldsymbol{\theta}_n\right)=&(\PPi'\Z'\Z\PPi+\PPi'\Z'\V+\V'\Z\PPi+\V'\PP_{Z}\V)^{-1}(\PPi'\Z'\uu+\V'\PP_{Z}\uu)\\&-(\PPi'\Z'\Z\PPi+\PPi'\Z'\V+\V'\Z\PPi+\V'\V)^{-1}(\PPi'\Z'\uu+\V'\uu),
\end{align*}
so that 
\begin{align*}
       \sqrt{n}\left(\hat{\bbeta}^{2SLS}-\hat{\bbeta}^{OLS}-\boldsymbol{f}(\boldsymbol{\theta})\right)\stackrel{d}{\rightarrow} \left( \nabla\boldsymbol{f}(\boldsymbol{\theta})\right) \boldsymbol{\xi}  \left( \nabla\boldsymbol{f}(\boldsymbol{\theta})\right)',
\end{align*}
where $$\boldsymbol{f}(\boldsymbol{\theta})=(\frac{\mathrm{E}(\PPi'\Z'\Z\PPi)}{K_n}+\boldsymbol{\Sigma}_{VV})^{-1}\boldsymbol{\Sigma}_{Vu}-(\frac{\mathrm{E}(\PPi'\Z'\Z\PPi)}{n}+\boldsymbol{\Sigma}_{VV})^{-1}\boldsymbol{\Sigma}_{Vu}.$$
Under $H_0$, $\nabla\boldsymbol{f}(\boldsymbol{\theta})=\left(\boldsymbol{g}_1,\boldsymbol{g}_2,\boldsymbol{g}_3,\boldsymbol{g}_4\right)$, where 
\begin{align}
    \boldsymbol{g}_1=\frac{\alpha-1}{\alpha}\big( \boldsymbol{\Sigma}_{VV}^{-1}\boldsymbol{J}^{11}\boldsymbol{\Sigma}_{VV}^{-1}\boldsymbol{\Sigma}_{Vu},\boldsymbol{\Sigma}_{VV}^{-1}\boldsymbol{J}^{12}\boldsymbol{\Sigma}_{VV}^{-1}\boldsymbol{\Sigma}_{Vu},\dots,\boldsymbol{\Sigma}_{VV}^{-1}\boldsymbol{J}^{pp}\boldsymbol{\Sigma}_{VV}^{-1}\boldsymbol{\Sigma}_{Vu} \big) \in \mathbb{R}^{p\times p^2},
\end{align}
\begin{align}
    \boldsymbol{g}_2=\frac{\alpha-1}{\alpha}\bigg( & \boldsymbol{\Sigma}_{VV}^{-1}(\boldsymbol{J}^{11}+\boldsymbol{J}^{11})\boldsymbol{\Sigma}_{VV}^{-1}\boldsymbol{\Sigma}_{Vu},\boldsymbol{\Sigma}_{VV}^{-1}(\boldsymbol{J}^{12}+\boldsymbol{J}^{21})\boldsymbol{\Sigma}_{VV}^{-1}\boldsymbol{\Sigma}_{Vu},\dots,\nonumber\\&\boldsymbol{\Sigma}_{VV}^{-1}(\boldsymbol{J}^{pp}+\boldsymbol{J}^{pp})\boldsymbol{\Sigma}_{VV}^{-1}\boldsymbol{\Sigma}_{Vu},-\boldsymbol{\Sigma}_{VV}^{-1}\boldsymbol{j}^{1},\dots,-\boldsymbol{\Sigma}_{VV}^{-1}\boldsymbol{j}^{p} \bigg) \in \mathbb{R}^{p\times (p^2+p)},
\end{align}
\begin{align} \label{g3}
    \boldsymbol{g}_3=\big(  \boldsymbol{\Sigma}_{VV}^{-1}\boldsymbol{J}^{11}\boldsymbol{\Sigma}_{VV}^{-1}&\boldsymbol{\Sigma}_{Vu},  \boldsymbol{\Sigma}_{VV}^{-1}\boldsymbol{J}^{12}\boldsymbol{\Sigma}_{VV}^{-1}\boldsymbol{\Sigma}_{Vu},\dots,\nonumber\\ &\boldsymbol{\Sigma}_{VV}^{-1}\boldsymbol{J}^{pp}\boldsymbol{\Sigma}_{VV}^{-1}\boldsymbol{\Sigma}_{Vu},-\boldsymbol{\Sigma}_{VV}^{-1}\boldsymbol{j}^{1},\dots,-\boldsymbol{\Sigma}_{VV}^{-1}\boldsymbol{j}^{p} \big) \in \mathbb{R}^{p\times (p^2+p)},
\end{align}
and
\begin{equation} \label{g4}
 \boldsymbol{g}_4= -\frac{1}{\alpha}\boldsymbol{g}_3.
\end{equation} Note that $\widetilde{\Z}_i$ is i.i.d by Assumption \ref{assum:2}(a), we have $\boldsymbol{\Sigma}_{1}=o(1)$, $\boldsymbol{\Sigma}_{2}=o(1)$ and $\boldsymbol{f}(\boldsymbol{\theta})=o_p(n^{-1/2})$. It yields that 
\begin{equation}
\sqrt{n}\left( \hat{\bbeta}^{2SLS}- \hat{\bbeta}^{OLS} \right)\stackrel{d}{\rightarrow}\tilde{\boldsymbol{\xi}},
\end{equation}
where $\tilde{\boldsymbol{\xi}}\in \mathbb{R}^{p}$ is a gaussian vector satisfying that $\mathrm{E}(\tilde{\boldsymbol{\xi}})=\boldsymbol{0}$, $\mathrm{Cov}(\tilde{\boldsymbol{\xi}})=-\boldsymbol{g}_3\boldsymbol{\Sigma}_{3}\boldsymbol{g}_3'+\frac{1}{\alpha^2}\boldsymbol{g}_3\boldsymbol{\Sigma}_{4}\boldsymbol{g}_3'$.
Under $H_1$, similarly,  $\nabla\boldsymbol{f}(\boldsymbol{\theta})=\left(\boldsymbol{h}_1,\boldsymbol{h}_2,\boldsymbol{h}_3,\boldsymbol{h}_4\right)$, where 
\begin{align} \label{h1}
    \boldsymbol{h}_1=\bigg(& \boldsymbol{\Theta}_2^{-1}\boldsymbol{J}^{11}\boldsymbol{\Theta}_2^{-1}\boldsymbol{\Sigma}_{Vu}-\alpha\boldsymbol{\Theta}_1^{-1}\boldsymbol{J}^{11}\boldsymbol{\Theta}_1^{-1}\boldsymbol{\Sigma}_{Vu},  \boldsymbol{\Theta}_2^{-1}\boldsymbol{J}^{12}\boldsymbol{\Theta}_2^{-1}\boldsymbol{\Sigma}_{Vu}-\alpha\boldsymbol{\Theta}_1^{-1}\boldsymbol{J}^{12}\boldsymbol{\Theta}_1^{-1}\boldsymbol{\Sigma}_{Vu},\nonumber\\&\dots, \boldsymbol{\Theta}_2^{-1}\boldsymbol{J}^{pp}\boldsymbol{\Theta}_2^{-1}\boldsymbol{\Sigma}_{Vu}-\alpha\boldsymbol{\Theta}_1^{-1}\boldsymbol{J}^{pp}\boldsymbol{\Theta}_1^{-1}\boldsymbol{\Sigma}_{Vu} \bigg) \in \mathbb{R}^{p\times p^2},
\end{align}
\begin{align}\label{h2}
    \boldsymbol{h}_2=\bigg(& \boldsymbol{\Theta}_2^{-1}(\boldsymbol{J}^{11}+\boldsymbol{J}^{11})\boldsymbol{\Theta}_2^{-1}\boldsymbol{\Sigma}_{Vu}-\alpha\boldsymbol{\Theta}_1^{-1}(\boldsymbol{J}^{11}+\boldsymbol{J}^{11})\boldsymbol{\Theta}_1^{-1}\boldsymbol{\Sigma}_{Vu}, \nonumber \\& \boldsymbol{\Theta}_2^{-1}(\boldsymbol{J}^{12}+\boldsymbol{J}^{21})\boldsymbol{\Theta}_2^{-1}\boldsymbol{\Sigma}_{Vu}-\alpha\boldsymbol{\Theta}_1^{-1}(\boldsymbol{J}^{12}+\boldsymbol{J}^{21})\boldsymbol{\Theta}_1^{-1}\boldsymbol{\Sigma}_{Vu},\nonumber\\&\dots,\nonumber\\&\boldsymbol{\Theta}_2^{-1}(\boldsymbol{J}^{pp}+\boldsymbol{J}^{pp})\boldsymbol{\Theta}_2^{-1}\boldsymbol{\Sigma}_{Vu}-\alpha\boldsymbol{\Theta}_1^{-1}(\boldsymbol{J}^{pp}+\boldsymbol{J}^{pp})\boldsymbol{\Theta}_1^{-1}\boldsymbol{\Sigma}_{Vu}, \nonumber\\& \left(\boldsymbol{\Theta}_1^{-1}-\boldsymbol{\Theta}_2^{-1}\right)\boldsymbol{j}^1, \dots, \left(\boldsymbol{\Theta}_1^{-1}-\boldsymbol{\Theta}_2^{-1}\right)\boldsymbol{j}^p\bigg) \in \mathbb{R}^{p\times (p^2+p)},
\end{align}
\begin{align}\label{h3}
    \boldsymbol{h}_3=\bigg(  \boldsymbol{\Theta}_2^{-1}\boldsymbol{J}^{11}\boldsymbol{\Theta}_2^{-1}\boldsymbol{\Sigma}_{Vu}, \boldsymbol{\Theta}_2^{-1}\boldsymbol{J}^{12}\boldsymbol{\Theta}_2^{-1}&\boldsymbol{\Sigma}_{Vu},\dots, \boldsymbol{\Theta}_2^{-1}\boldsymbol{J}^{pp}\boldsymbol{\Theta}_2^{-1}\boldsymbol{\Sigma}_{Vu},\\&-\boldsymbol{\Theta}_2^{-1}\boldsymbol{j}^{1},\dots,-\boldsymbol{\Theta}_2^{-1}\boldsymbol{j}^{p} \bigg)  
    \in \mathbb{R}^{p\times (p^2+p)\nonumber}
\end{align}
and
\begin{align}\label{h4}
    \boldsymbol{h}_4=\bigg(  \alpha\boldsymbol{\Theta}_1^{-1}\boldsymbol{J}^{11}\boldsymbol{\Theta}_1^{-1}\boldsymbol{\Sigma}_{Vu}, \alpha\boldsymbol{\Theta}_1^{-1}\boldsymbol{J}^{12}\boldsymbol{\Theta}_1^{-1}&\boldsymbol{\Sigma}_{Vu},\dots, \nonumber \alpha\boldsymbol{\Theta}_1^{-1}\boldsymbol{J}^{pp}\boldsymbol{\Theta}_1^{-1}\boldsymbol{\Sigma}_{Vu},\\&-\boldsymbol{\Theta}_1^{-1}\boldsymbol{j}^{1},\dots,-\boldsymbol{\Theta}_1^{-1}\boldsymbol{j}^{p} \bigg) \in \mathbb{R}^{p\times (p^2+p)}.
\end{align}
It yields that 
\begin{equation*}
\sqrt{n}\left( \hat{\bbeta}^{2SLS}- \hat{\bbeta}^{OLS}-\boldsymbol{\Delta} \right)\stackrel{d}{\rightarrow}\Breve{\boldsymbol{\xi}},
\end{equation*}
where $\Breve{\boldsymbol{\xi}}\in \mathbb{R}^{p}$ is a gaussian vector satisfying that $\mathrm{E}(\Breve{\boldsymbol{\xi}})=\boldsymbol{0}$, $\mathrm{Cov}(\tilde{\boldsymbol{\xi}})=\boldsymbol{h}_1\boldsymbol{\Sigma}_{1}\boldsymbol{h}_1'+\boldsymbol{h}_2\boldsymbol{\Sigma}_{2}\boldsymbol{h}_2'+\boldsymbol{h}_3\boldsymbol{\Sigma}_{3}\boldsymbol{h}_3'+\alpha\boldsymbol{h}_3\boldsymbol{\Sigma}_{3}\boldsymbol{h}_4'+\alpha\boldsymbol{h}_4\boldsymbol{\Sigma}_{3}\boldsymbol{h}_3'+\boldsymbol{h}_4\boldsymbol{\Sigma}_{4}\boldsymbol{h}_4'$.
\end{proof}

\begin{proof}[Proof of Corollary \ref{cor:p1}]
When $p=1$,
\begin{equation*}
    g_1=\frac{\alpha-1}{\alpha}\frac{\sigma_{Vu}}{\sigma_{VV}^4}, 
\end{equation*}
\begin{equation*}
    \boldsymbol{g}_2=\frac{\alpha-1}{\alpha}(2\frac{\sigma_{Vu}}{\sigma_{VV}^4},-\frac{1}  \nonumber{\sigma_{VV}^2}) 
\end{equation*}
and
\begin{equation*}
    \boldsymbol{g}_3=(\frac{\sigma_{Vu}}{\sigma_{VV}^4},-\frac{1}{\sigma_{VV}^2}). 
\end{equation*}
By Isserlis' theorem, we verify that 
\begin{equation*}
   \mathrm{Var}(v_1^2)=2\sigma_{VV}^4, \quad \mathrm{E}(v_1^3u_1)=3\sigma_{VV}^2\sigma_{Vu},  
\end{equation*}
which yields that
\begin{equation*}
    \boldsymbol{\Sigma}_{3}= \begin{bmatrix}
      2\sigma_{VV}^4 &2\sigma_{VV}^2\sigma_{Vu}\\
      2\sigma_{VV}^2\sigma_{Vu} & \sigma_{VV}^2\sigma_{u}^2+\sigma_{Vu}^2
    \end{bmatrix}\quad  \text{and} \quad \boldsymbol{\Sigma}_{4}= \alpha   \boldsymbol{\Sigma}_{3}.
\end{equation*}
Finally, we have 
\begin{equation*}
  \sigma^2=\frac{1-\alpha}{\alpha}(\frac{\sigma_{u}^2}{\sigma_{vv}^2}-\frac{\sigma_{vu}^2}{\sigma_{vv}^4}). 
\end{equation*}
\end{proof}

\subsection{Proof of Theorem  \ref{thm:jsve}}
\begin{proof}[Proof of Theorem \ref{thm:jsve}]
It can be verified that 
\begin{equation}
    \hat{\bbeta}^{2SLS}-\hat{\bbeta}^{OLS}=(\Y'\PP_Z\Y)^{-1}\Y'\PP_Z\mathbf{M}_Y\y:=\Tilde{\Y}'\y,
\end{equation}
where $\Tilde{\Y}=\mathbf{M}_Y\PP_Z\Y(\Y'\PP_Z\Y)^{-1}$. Without loss of generality, we assume $p=1$. Define $\Tilde{\y}=\y-\mathrm{E}(\Tilde{Y}_1y_1)/\mathrm{Var}(\Tilde{Y}_1)\Tilde{\Y}$. Note that
\begin{equation}
    \hat{\beta}^{2SLS}-\hat{\beta}^{OLS}=\Tilde{\Y}'\Tilde{\y}+\frac{\mathrm{E}(\Tilde{Y}_1y_1)}{\mathrm{Var}(\Tilde{Y}_1)}\Tilde{\Y}'\Tilde{\Y}=\frac{1}{n}\sum_{i=1}^n\phi_i+R_n,
\end{equation}
where $\phi_i=n\Tilde{Y}_iy_i$ with zero mean and variance $\sigma_0^2$ and $R_n=\mathrm{E}(\Tilde{Y}_1y_1)/\mathrm{Var}(\Tilde{Y}_1)\Tilde{\Y}'\Tilde{\Y}$. We need to verify the condition (2.1) in Shao (1989) under $H_0$, namely, $\mathrm{E}R_n^2=o(n^{-1})$. If suffices to show that
\begin{equation*}
   \mathrm{E}\left(\frac{1}{\Y'\PP_Z\Y}-\frac{1}{\Y'\Y}\right)^2=o(n^{-1}).
\end{equation*}
It then holds by Assumption \ref{assum:bound} as $\Y'\PP_Z\Y\leq \Y'\Y$.
\end{proof}

\bibliographystyle{chicago}

\bibliography{Bibliography-MM-MC}

\end{document}